\documentclass[11pt]{article}

\date{}

\usepackage{url}
\usepackage{cite}
\usepackage{plain}

\usepackage{wrapfig}
\usepackage{graphics}
\usepackage{picins,graphicx}
\usepackage[english]{babel}
\usepackage[leqno]{amsmath}
\usepackage{amssymb}
\usepackage{verbatim}
\usepackage[mathscr]{eucal}

\usepackage{amstext}
\usepackage{makeidx}
\usepackage{amsthm}

\usepackage{hyperref}
\hypersetup{colorlinks,%
citecolor=red,%
linkcolor=blue,%
urlcolor=black}

\makeindex

\newtheorem{theorem}{Theorem} 

\newcommand{\mail}[1]{\href{unina:#1}{\texttt{#1}}}

\usepackage{pdfsync}

\author{Monica De Angelis\thanks{University of Naples  "Federico II", Dep. Mat. Appl. "R.Caccioppoli", \newline
 Via Claudio n.21, 80125, Naples, Italy.
\newline\mail{modeange@unina.it}}}

\title {Asymptotic estimates  related to  an  integro differential equation}

\begin{document}
\maketitle

\begin{abstract}
 The paper deals with an   integrodifferential operator   which  models   numerous phenomena  in superconductivity, in biology and in viscoelasticity. Initial-boundary value problems with Neumann, Dirichlet and mixed boundary conditions are analyzed. An  asymptotic analysis    is  achieved proving that for large t, the influences of the initial data vanish, while the effects of boundary disturbances are everywhere bounded.

\vspace{5mm}{{Keywords}}:{Initial- boundary problems for higher order parabolic equations;\, Laplace  transform;\, Superconductivity;\,FitzHugh Nagumo model.}

\vspace{5mm} \textbf{Mathematics Subject Classification (2000)}44A10,\,35K57,\,35A08,\,35K35
\end{abstract}


\vspace{13mm}

\bigskip
\section{Introduction}

If $ u=u(x,t), $ let us consider  the  following  integrodifferential equation

\begin{equation}   \label{11}
  \mathcal L u \equiv \,\, u_t -  \varepsilon  u_{xx} + au +b \int^t_0  e^{- \beta (t-\tau)}\, u(x,\tau) \, d\tau \,=\, F(x,t,u) \, 
  \end{equation}  
  
  \noindent where $ \varepsilon, a, b, \beta  $ are positive constants, $ x $ denotes the  direction of propagation and $ t $ is the time. According to the meaning of $ F(x,t,u)$, equation (\ref{11})  describes   the evolution of several   linear or non linear  physical models. For instance, when $ F=f(x,t), $ (\ref{11}) is related to the following linear phenomena:
  
\begin{itemize}
\item motions of viscoelastic fluids or solids  \cite{bcf,dr1,dmr,mor};
\item heat conduction at low temperature \cite{mps,fr,bs},
\item sound propagation in viscous gases \cite{l}.\end{itemize}   

 \noindent When $ F=F(x,t,u), $ some  non linear phenomena involve equation (\ref{11}) both in superconductivity and  biology.

\vspace{3mm}
 $ \bullet $  {\em Superconductivity} --  Let  $\, u\, $ be the difference between the wave functions phases of  two superconductors in a  Josephson junction. The equation describing   tunnel effects is the following one:   

\begin{equation}  \label{122}
\varepsilon u_{xxt}\, - \, u_{tt} \, +\, u_{xx}-\, \alpha u_t = \,  \, \sin u \ - \gamma   
\end{equation}

\noindent where constant $\, \gamma \, $ is a forcing term  proportional 
to a bias current,   while  the  $ \varepsilon -term$  and the $\, \alpha -term $ account for  the dissipative normal electron current  flow, respectively along and across the junction  \cite{bp,df13}.

Equation (\ref{122})  can be obtained by (\ref{11})  as soon as one assumes

\begin{equation}   \label{133}
 a \,=\,  \alpha \, - \dfrac{1}{\varepsilon} \, \,\quad\,\, b = \,  - \, \dfrac{a}{\varepsilon}  \,\,\quad \displaystyle \, \beta \,= \dfrac{1}{\varepsilon}\,\,
  \end{equation}
  
  \noindent  and $ F  $ is such that

\begin{equation}   \label{144}
F(x,t,u)\,=\, -\, \int _0^t \, e^{\,-\,\frac{1}{\varepsilon}\,(t-\tau\,)}\,\,[\, \, sen \, u (x, \tau)\,- \gamma\, \,]\, \, d\tau. 
\end{equation}

\noindent  Besides, when the case of an exponentially shaped Josephson junction (ESJJ)  is considered, the evolution of the phase inside this junction is described  by the third order equation:

\begin{equation}  \label{12}
(\partial_{xx} \, - \,\lambda\,  \partial _x\,)\,\,(\varepsilon
u _{t}+ u) - \partial_t(u_{t}+\alpha\,u)\, = \,  \, \sin u \ - \gamma   
\end{equation}

\noindent where  $ \lambda $ is a positive constant generally less than one    and  the terms $ \lambda u_{xt} $           and  $ \lambda u_{x} $ represent the current due to the  tapering junction. In particular  $ \lambda u_{x} $ corresponds  to a geometrical force driving the fluxons from the wide edge to the narrow edge.   
  \cite{df13,df213,bcs}  An   (ESJJ)   provides several advantages with respect to a rectangular junction (\cite{mda13} and reference  therein). For instance, in \cite{bcs} it has been proved that  it is possible to obtain  a voltage which is not  chaotic anymore, but rather periodic  excluding, in this way,   some among the possible causes of large spectral width. It is also proved that  the problem of trapped flux can be avoided. Numerous applications  and devices  involve Josephson  junctions,  for  example SQUIDs   which are very versatile and  can be used in  a lot of  fields. (see f.i.\cite{cd} and references  therein).
  
 \vspace{3mm} Moreover, if  $\,\,u= e^{\lambda\,x/2 \,}\,\,\overline{u},\,\, $  (\ref{12}) turns into an equation like  (\ref{122}) and  hence into  (\ref{11}).

\vspace{3mm}
$ \bullet $ {\em Biology} -- 
 Let us consider the FitzHugh-Nagumo system  (FHN) which models the propagation of nerve impulses. \cite{m1}:

 \begin{equation}     \label{15}
  \left \{
   \begin{array}{lll}
    \displaystyle{\frac{\partial \,u }{\partial \,t }} =\,  \varepsilon \,\frac{\partial^2 \,u }{\partial \,x^2 }
     \,-\, v\,\,  + f(u ) \,  \\
\\
\displaystyle{\frac{\partial \,v }{\partial \,t } }\, = \, b\, u\,
- \beta\, v\,.
\\

   \end{array}
  \right.
 \end{equation}

\noindent Here,   $\, u\,( \,x,t\,)$  models the transmembrane voltage  of a nerve axon at  a distance x and  time t, while  $\, v\,(\,x,t\,)$ is an auxiliary variable   acting  as  a recovery variable. 
Besides, the  function $\,f (u)\,$  has the qualitative form of a cubic polynomial

\begin{equation}      \label{16}
f(u)\, =\,-\,a\,u\, +\,\varphi(u) \quad with \quad  \varphi \,=\, u^2\, (\,a+1\,-u\,),
\end{equation}

\noindent while  $ \varepsilon,\, b,\, \beta\, $ are non negative and  the  parameter $ a, $  representing the threshold constant, is generally   $ \,0<a<1.$ (see f.i. \cite{rda08} and references therein)

Denoting by  $\,v_0  \,$  the initial value of v,  system (\ref{15}) (\ref{16}) can be given the form of the integrodifferential equation  (\ref{11})   as soon as one puts: 

\begin{equation}     \label{51}
F(x,t,u)\, =\,\varphi (u) \, -\, v_0(x) \, e^{\,-\,\beta\,t\,}.
\end{equation} 

\vspace{3mm} In this paper, initial value problems with  Neumann,  Dirichlet and mixed boundary conditions for (\ref{11}) are  considered. By means of properties of  the fundamental solution $ K_0(x,t,)$ of the operator $ {\cal L} ,$   appropriate  estimates are obtained.     The function  $ K_0(x,t)$    has already been determined and analyzed 
in \cite{dr8} and  an analysis related to a Neumann boundary problem has been conducted   in
\cite{dr13}. Aim of this paper  is an asymptotic analysis  for the  initial boundary  value problem  both with Dirichlet conditions and with mixed conditions. These cases involve  x-derivative of   theta functions $ \theta(x,t) $ and $ \theta^*(x,t) $ which are determined in sec (\ref{sec3}). So, effects of boundary perturbations can be  evaluated by means of a well known theorem on asymptotic behavior of convolutions.  As  an example,  according to the  equivalence between operator $ \mathcal{L} $  and  the FHN system,  an  estimate of the  solution related  to the reaction-diffusion system (\ref{15}) is  obtained proving that,  for large $ t, $ effects determined by boundary disturbance  are bounded.   

\section{Some models of superconductivity and biology}

Let  $\, T\, $  be   an arbitrary positive constant and 

\[
\,   \Omega_T \, \equiv \{\,(x,t) : \, 0\,\leq \,x \,\leq L \,\,;  \ 0 < t \leq T. \, \]

 \vspace{2mm} (I) A first example is related to {\it Neumann} boundary conditions (NBC)

\begin{equation}   \label{21}
\left \{
   \begin{array}{lll}
{  \cal L}\,u\,  \,=\, F(x,t,u) \, & (x,t) \in \Omega_T \,  \\    
\\  \,u (x,0)\, = u_0(x)\, \,\,\, &
x\, \in [0,L], 
\\
\\
  \, u_x(0,t)\,=\,\psi_1(t)  \qquad u_x(L,t)\,=\,\psi_2(t) & 0<t\leq T.
   \end{array}
  \right.
\end{equation}

\noindent In superconductivity, this problem occurs 
when the magnetic field, proportional to the phase gradient, is assigned. \cite{j,ddf}. In mathematical biology, it can refer to   a two-species reaction diffusion system subjected to flux boundary conditions \cite{m1}.  The same conditions  are present in case of pacemakers \cite{ks} and are applied also to study distributed (FHN) systems \cite{ns} or to solve FHN systems by means of  numerical calculations \cite{d}.


 \vspace{3mm}(II) Another example concerns {\it Dirichlet} boundary conditions  (DBC)

 \begin{equation}   \label{22}
\left \{
   \begin{array}{lll}
 {  \cal L}\,u\,  \,=\, F(x,t,u) \, & (x,t) \in \Omega_T \,  \\     \\
  \,u (x,0)\, = u_0(x)\, \,\,\, &
x\, \in [0,L], 
\\\\
  \,u(0,t)\,=\,g_1(t)  \qquad u(L,t)\,=\,g_2(t)    & 0<t\leq T.
   \end{array}
  \right.
\end{equation}

 \noindent In superconductivity,  $(\ref{22})_3$  refer to the phase boundary specifications\cite{df213,df13,mda13}. In excitable systems these conditions  occur when  the behavior  of a single dendrite has to be  determined and  the voltage level is fixed\cite{ks} or when  the  pulse propagation in
a continuum of heart cells is studied \cite{a,ks}. Besides, the Dirichlet problem  is also considered to determine  universal attractors both for Hodgkin-Huxley equations and for  FHN systems,\cite{m} and for stability  analysis and asymptotic behavior    of reaction-diffusion systems   solutions,   \cite{dm13,t,ccd,ra,mda14}, or  in hyperbolic diffusion \cite{gs}.

\vspace{3mm}(III) At last, {\it mixed} boundary conditions   (MBC) as

 \begin{equation}   \label{23}
\left \{
   \begin{array}{lll}
   {  \cal L}\,u\,  \,=\, F(x,t,u) \, & (x,t) \in \Omega_T \,  \\     
\\
  \,u (x,0)\, = u_0(x)\, \,\,\, &
x\, \in [0,L], 
\\
\\ 
  \,u(0,t)\,=\,h_1(t)  \qquad u_x(L,t)\,=\,h_2(t)    & 0<t\leq T,
   \end{array}
  \right.
\end{equation} 
 
 \noindent  occur in  many physical examples both in superconductivity (see,f.i.\cite{bv} and references therein) and in biology,  as shown in \cite{m1,ks}. In particular, in \cite{rs}, mixed boundary conditions are considered in order to  give qualitative information concerning both the threshold
problem and the asymptotic behavior of large solutions for the  FHN system.

\vspace{5mm} When   $ F\,= f(x,t) $ is a linear function,  problems  (\ref{21})-(\ref{23})   can  be solved by Laplace transformation  with respect to $ t. $

\noindent Let $ z(x,t) $ be  an arbitrary function  admitting Laplace transform $ \hat z(x,s) $

 \begin{equation}   \label{24}
\hat z (x,s) \, = \int_ 0^\infty \, e^{-st} \, z(x,t) \,dt \,\,= \mathcal{L}_t\,z
\end{equation}

\noindent Referring to the parameters $ a, \,\beta,\, b, \,\,\varepsilon\,\, $ of the operator $ \mathcal{L},  $   if

\begin{equation} \label{25}
 \sigma^2 \ \,=\, s\, +\, a \, + \, \frac{b}{s+\beta},\,\quad \,\,\displaystyle\,\tilde{\sigma}^2\,=\, \sigma^2/{\varepsilon,}\,\,\end{equation}
 
 \noindent    we denote by $ \theta(x,s) $ and $ \theta ^* (x,s) $ the following Laplace  transforms:

\begin{equation}\,  \label{26}
\displaystyle
\hat \theta \,(\,y,\tilde\sigma)\,= \,\dfrac{\cosh\,[\, \tilde\sigma \,\,(L-y)\,]}{\,2\, \, {\varepsilon} \,\,\tilde \sigma\,\,\, \sinh\, (\,\tilde \sigma \,L\,)}\,\,=
\end{equation}

\[ =\,  \frac{1}{2 \,\, \sqrt\varepsilon \,\,\,\sigma  } \, \biggl\{\, e^{- \frac{y}{\sqrt \varepsilon} \,\,\sigma}+\, \sum_{n=1}^\infty \,\, \biggl[ \,e^{- \frac{2nL+y}{\sqrt \varepsilon} \,\,\sigma} \, +\, e^{- \frac{2nL-y}{\sqrt \varepsilon} \,\,\sigma}\,
\biggr] \, \biggr\},    \]

\begin{equation}\,  \label{27}
\displaystyle
\hat \theta^* \,(\,y,\tilde\sigma)\,= \,\dfrac{\sinh\,[\, \tilde\sigma \,\,(L-y)\,]}{\,2\, \, {\varepsilon} \,\,\tilde \sigma\,\,\, \cosh\, (\,\tilde \sigma \,L\,)}\,\,=
\end{equation}

\[ =\,  \frac{1}{2  \sqrt\varepsilon \sigma  } \, \biggl\{     e^{- \frac{y}{\sqrt \varepsilon}\sigma}+\,\, 2\,\,\sum_{n=1}^\infty \,\, \biggl( \,e^{- \frac{4nL+y}{\sqrt \varepsilon} \,\,\sigma} \, +\, e^{- \frac{4nL-y}{\sqrt \varepsilon} \,\,\sigma}\,
\biggr)  \,    -  \sum_{n=1}^\infty \,\, \biggl( \,e^{- \frac{2nL+y}{\sqrt \varepsilon} \,\,\sigma} \, +\, e^{- \frac{2nL-y}{\sqrt \varepsilon} \,\,\sigma}\,
\biggr)   \biggr\}.    \]

\vspace{3mm}\noindent Then, the Laplace transform solutions  of the linear problems (\ref{21})-(\ref{23})
 can be obtained by means of standard techniques   and it results:

\vspace{3mm}

$ \bullet $  
Formal solution for  initial boundary problem with (NBC)

\begin{equation}     \label{28}
\begin{split}
\hat u (x,s) = &\,\int _0^L \, [\,\hat \theta\,(\,|x-\xi|, \,s\,)\,+\,\,\,\hat \theta\,(\,|x+\xi|,\, s\,)\,] \, \,[\,u_0(\,\xi\,) \,+\,\hat f(\,\xi,s)\,]\,d\xi\,
\\ 
  \,&  -\,\,\ 2 \,\,\varepsilon \, \,\hat \psi_1 \,(s) \,\, \hat  \theta (x,s)\,+ \, 2 \,\, \varepsilon  \,\, \hat \psi_2 \, (s)\,\,\hat  \theta \,(x-L,s\,).\, \,\,
\end{split}
\end{equation}

$ \bullet $ 
Formal solution for  (DBC)

\begin{equation}     \label{29}
\begin{split}
\hat u (x,s) = &\,\int _0^L \, [\,\hat \theta\,(\,|x-\xi|, \,s\,)\,-\,\,\,\hat \theta\,(\,x+\xi,\, s\,)\,] \, \,[\,u_0(\,\xi\,) \,+\,\hat f(\,\xi,s)\,]\,d\xi\, -
\\ 
  \, & -\,\,\ 2 \,\,\varepsilon \, \,\hat g_1 \,(s) \,\,    \hat\theta_x (x,s)\,+ \, 2 \,\, \varepsilon  \,\, \hat g_2 \, (s)\,\,\,\hat  \theta_x \,(x-L,s\,).\, \,\,
\end{split}
\end{equation}

$ \bullet $    Formal solution  for (MBC)

\begin{equation}     \label{210}
\begin{split}
\hat u (x,s) = &\,\int _0^L \, [\,\hat \theta^*\,(\,x+\xi, \,s\,)\,-\,\,\,\hat \theta^*\,(\,|x-\xi|,\, s\,)\,] \, \,[\,u_0(\,\xi\,) \,+\,\hat f(\,\xi,s)\,]\,d\xi\, +
\\ 
  \,& \,-\,\ 2 \,\,\varepsilon \,  \,\hat h_1 \,(s) \,\, \,\,\hat  \theta^*_x (x,s)\,+ \, 2 \,\,\varepsilon  \hat h _2 \, (s)\,\,\hat  \theta^* \,(L-x,s\,) . \,\,
\end{split}
\end{equation}

\section{  $K_0 (x,t)$  and $ \theta(x,t) $ properties} \label{sec3}

 The Neumann boundary value  problem has  already been  solved in 
\cite{dr13}. Let us consider now  cases (II) and (III).

  Let $ K_0(x,t)  $ be  the fundamental solution of the linear operator  $\, \cal L \,$ defined in (\ref{11}). It has already been determined in \cite{dr8} and one has:

\begin{equation}  \label{31}
K_0(r,t)=  \frac{1}{2 \sqrt{\pi  \varepsilon } }\biggl[ \frac{ e^{- \frac{r^2 }{4 t}-a\,t}}{\sqrt t}-\,\sqrt b \int^t_0  \frac{e^{- \frac{r^2}{4 y}\,- ay}}{\sqrt{t-y}} \, e^{-\beta (\, t \,-\,y\,)}  J_1 (2 \sqrt{\,by\,(t-y)\,})\,\,dy \biggr]
\end{equation} 

\noindent where $\, r \, = |x| \, / \sqrt \varepsilon \, \, $  and  $ J_n (z) \,$  is  the Bessel function of first kind. Function   $ K_0 \, $  has the same  basic properties of the fundamental solution of the heat equation, and  in the half-plane $ \Re e  \,s > \,max(\,-\,a ,\,-\beta\,)\,$ it results:

\begin{equation}      \label{32}
\,{\cal L  }_t\,\,K_0\,\equiv \,\,\int_ 0^\infty e^{-st} \,\, K_0\,(r,t) \,\,dt \,\,=  \,
 \frac{e^{- \,r\,\sigma}}{2 \, \sqrt\varepsilon \,\sigma \,  } 
\end{equation}

\noindent  where $ \sigma  $   is defined in  $(\ref{25})_1$.

   \vspace{3mm} Among other properties, in \cite{dr8} the following estimates have been proved:

\begin{equation}               \label{33}
\int_\Re|\,K_0(x-\xi,t)|d\xi\leq  e^{\,-\,at} +\, \sqrt b\, \pi t  e^{\,-\,\omega \, t } \quad \int_0^t\,d \tau\, \int_\Re |K_0(x-\xi,t)| \, \,d\xi \leq \,  \beta_0 
\end{equation} 
\begin{equation}               \label{34}
|K_0| \, \leq \, \frac{e^{- \frac{r^2}{4 t}\,}}{2\,\sqrt{\pi \varepsilon t}} \,\, [ \, e^{\,-\,at}\, +\, b t \,E(t)\, ] 
\end{equation}

\vspace{3mm}\noindent where  constants $ \omega, \,\,\beta_0\,\, \mbox {and}\,\, E(t) $ are given by:

\begin{equation}      \label{35} 
  \omega = min(a,\beta),\,\,\qquad   \beta _0 =\,\, \frac{1}{a}\, +\, \pi \sqrt b \, \, \displaystyle {\frac{a+\beta}{2(a\beta)^{3/2}}},\end{equation}
   \[
 E(t) \,=\, \frac{e^{\,-\,\beta t}\,-\,e^{\,-\,at}}{a\,-\,\beta}\,\,>0.\,\]

\noindent Moreover,  denoting by

\begin{equation}     \label {36}
 K_i( r,t) \, = \,\,\int^t_0 \,\,e^{-\,\beta \,(\,t-\tau)\,}\,K_{i-1}\,(x,\tau\,) \, d \tau\,\, \qquad ( i=1,2) 
 \end{equation}

 \noindent  kernels $ K_1(x,t)  $ and $ K_2(x,t)  $ have the same  properties of $ K_0(x,t).  $ Hence,  the following theorem holds \cite{dr8}:

\begin{theorem}

For all the positive constants $ a, \,b,\,\, \varepsilon,\, \beta $  it results:

\begin{equation}   \label{37}
\int_\Re |K_1| \, \ d\xi \leq \, \,E(t);\,  \qquad  \int _0^t\\d\tau \,\int_\Re |K_1| \, \ d\xi \leq \, \beta_1\,
\end{equation}

\begin{equation}   \label{38}
\int_\Re \left|K_2 (x-\xi,t)\right| \, d\xi \, \leq \, t\, E(t)
\end{equation}

 \noindent where $   \beta_1\,=\, ({a\,\beta})^{\,-1}.\,$ 

\end{theorem}



So that, in order to obtain   inverse formulae of (\ref{29}) and  (\ref{210}),        let  us apply (\ref{32}) to (\ref{26})(\ref{27}). Then,  one deduces the  following functions which are  similar to  {\em theta functions}:

\begin{equation}     \label{310}
\begin{split}
\theta (x,t) \,=\,&  K_0(x,t) \ +\, \sum_{n=1}^\infty \,\, \ [\, K_0(x \,+2nL,\,t) \, + \, K_0 ( x-2nL, \,t)\,] \,  \\&= \sum_{n=-\infty }^\infty \,\, \ K_0(x \,+2nL,\,t).
\end{split} 
\end{equation}

\begin{equation}     \label{3100}
\begin{split}
\theta^* (x,t) \,= \,2\,\sum_{n=-\infty }^\infty \,  K_0(x \,+4nL,\,t)\, - \,\sum_{n=-\infty }^\infty \,\,  K_0(x \,+2nL,\,t).
\end{split} 
\end{equation}
  
Some of the properties of function $ \theta(x,t) $ have  already been evaluated in \cite{dr13}. Precisely, denoting by    $C= 2 \varepsilon \,\,\pi^2 / (\, 6\,e  L ^2\,) $    and  letting

\begin{equation}   \label{311}
 C_0\, = \, \frac{1}{ 2 \sqrt{\varepsilon \, \omega }}\,+\,\frac{ b\,\,\omega^{-\, 3/2}}{4 \sqrt{ \varepsilon } \,\,|a-\beta|} \,\biggl[\,1\,\, +\, \dfrac{C}{b}\,|a-\beta|\,+ \frac{3 \,C}{2\, \omega }\,\,\biggr],\end{equation}

\noindent the  $\theta (x,t)$ function,  defined in $( \ref{310}), $ satisfies the following inequalities:

\begin{equation}               \label{312}
 \int_0^L |\theta (|x-\xi|,\,t)|\ \, d\xi \leq \,  ( 1\, +\, \sqrt b \,\pi \,t \, ) \,\,e^{- \omega \, t\,} 
\end{equation}

\begin{equation}               \label{313}
   \int_0^t\,d \tau\, \int_0^L |\theta (|x-\xi|,\,t)|\ \, d\xi \leq \,  \beta_0;\,\quad \quad \int_0^ \infty \, | \theta ( x,\tau )| \,\, d \tau  \,\, \leq \,\,C_0, 
\end{equation}

\vspace{3mm}\noindent and, it results:

\begin{equation}              \label{314}
\lim _{t \to \infty}  \theta ( x, t ) \,\, = \,\,0;\qquad\lim _{t \to \infty} \int _0^t \theta ( x, \tau ) \,\, d\tau \,\,= \frac{1}{2 \, \varepsilon \,\,\sigma _0\,  }\,\,\ \frac{\cosh \sigma_0 \,\,(L-x)}{\sinh\,(  \sigma_0 \, L). }
\end{equation}

\noindent  where $ \sigma_0 = \sqrt{\biggl(\,a\,\,+ \dfrac{b}{\beta}\biggr)\dfrac{1}{\varepsilon}}.$

\vspace{3mm} \noindent Furthermore, as for $ \frac{\partial \theta}{\partial x } $, from (\ref{31}), it is well-rendered  that the x derivative of the  integral term vanishes  for $ x\,\rightarrow 0 \,$, while the first term represents  the  derivative with respect to  $ x $ of  the fundamental solution related to the heat equation. So, by means of  classic theorems (see,f.i. \cite{c} p. 60),  conditions $(\ref{22})_3$ are surely  satisfied.  

\vspace{3mm}Moreover  one has:  
\begin{equation}              \label{317}
\begin{split}
\lim _{t \to \infty} \int _0^t \theta_x ( x, \tau ) \,\, d\tau \,\,= \frac{1}{2 \, \varepsilon \,  }\,\,\ \frac{\sinh \sigma_0 \,\,(x-L)}{\sinh\,(  \sigma_0 \, L) }\\  \\ \lim _{t \to \infty} \int _0^t \theta^*_x ( x, \tau ) \,\, d\tau \,\,= \,-\,\frac{1}{2 \, \varepsilon \,  }\,\,\ \frac{\cosh \sigma_0 \,\,(L-x)}{\cosh\,(  \sigma_0 \, L) }\end{split}
\end{equation}

\section{Asymptotic behaviour} \label{sec4}

When the source term $ F =f(x,t) $  is  a prefixed function depending only on $ x $ and $ t $, 
then,  initial boundary value problems (\ref{22}) (\ref{23}) are linear and  can be solved explicitly. Moreover, when    $\, F\, =\, F(x,t,u) \,$   depends also on the unknown function $ u(x,t), $  then these problems   admit  integral differential formulations and one has:

\vspace{3mm}$ \bullet $   Integro differential equation for  problem (\ref{22}) (DBC):

\begin{equation}   \label{41}
\begin{split}
 u(\, x,\,t\,)\, = \,\,\int^L_0 \, [\theta \,(|x-\xi|,\, t)\,- \theta (x+\xi,\,t)\,]\, \,u_0(\xi)\,\, d\xi \,\,- \,\\ \\
\,2 \, \varepsilon \,\int^t_0 \theta_x\, (x,\, t-\tau) \,\,\, g_1 (\tau )\,\,d\tau\,+\, 2\,\, \varepsilon \int^t_0 \theta_x\, (x-L,\, t-\tau) \,\,\, g_2 (\tau )\,\,d\tau\,
\\\\ +\,\int^t_0 d\tau\int^L_0 \, [\,\theta\, (|x-\xi|,\, t-\tau)- \theta (x+\xi,\,t-\tau )] \,\,\, F\,(\,\xi,\tau,\,u(x,\tau))\, \,\,d\xi.
 \end{split}
\end{equation}

 $ \bullet $   Integro differential equation for  (\ref{23}) (MBC):

\begin{equation}   \label{42}
\begin{split}
 u(\, x,\,t\,)\, = \,\,\int^L_0 \, [\theta^* \,(|x-\xi|,\, t)\,- \theta^* (x+\xi,\,t)\,]\, \,u_0(\xi)\,\, d\xi \,\,- \,\\ \\
\,2 \, \varepsilon \,\int^t_0 \theta^*_x\, (x,\, t-\tau) \,\,\, h_1 (\tau )\,\,d\tau\,+\, 2\,\, \varepsilon \int^t_0 \theta^*\, (L-x,\, t-\tau) \,\,\, h_2 (\tau )\,\,d\tau\,
\\\\ +\,\int^t_0 d\tau\int^L_0 \, [\,\theta^*\, (|x-\xi|,\, t-\tau)- \theta^* (x+\xi,\,t-\tau )] \,\,\, F\,(\,\xi,\tau,\,u(x,\tau))\, \,\,d\xi.
 \end{split}
\end{equation}

   \vspace{3mm} 
   Now, if  $ \,{\cal B}_ T \,  $ denotes the Banach space 

\begin{equation}   \label{43}
  \,{\cal B}_ T \, \equiv \, \,\bigg\{\, z\,(\,x,t\,) : \, z\, \in  C \,(\Omega_T),  \, \,\,   ||\,z\,|| \,= \displaystyle \sup _{ \Omega_T\,}\, | \, z \,(\,x,\,t) \,|, \,\, < \infty \bigg \}
\end{equation}

 \noindent and  $ D $  is the following set:  
 \[ D\equiv  \{(x,t,u) : (x,t) \in \Omega_T , -\infty < u <\infty\,  \]

\noindent then, let  assume the source term $ F(x,t,u) \, $  be defined and continuous on $ D $ and   uniformly Lipschitz continuous in $(x,t,u)$ for each compact subset of $ \Omega_T.  $  Besides, let $ F $  be a bounded function  for bounded $ u $ and there exists a constant $ C  $ such that:

\[ |F(x,t,u_1)-F(x,t,u_2)|  \leq\,\, C\,\, |u_1-u_2|.\]

 \noindent So, by means of standard methods related to integral equations  and owing to basic properties of $ K_0,$ it is  possible to  prove that  the mappings defined by (\ref{41}) (\ref{42})   are  a contraction  of $ {\cal B}_ T $ in $ {\cal B}_ T  $ and so they admits a unique fixed point   $ u(x,t)  \, \in {\cal B}_ T $. \cite{c,dmm}  

  \vspace{3mm} In order to enable a quicker reading, attention will be paid  only to the initial boundary  value problem with Dirichlet conditions. However,  all the  following analysis can be applied to the mixed problem,too.

\vspace{3mm} At first, let us consider $ g_i =0\,\,(i=1,2) $ and let

\[\,\,\,||\,u_0\,|| \,= \displaystyle \sup _{ 0\leq\,x\,\leq \,L\,}\, | \,u_0 \,(\,x\,) \,|, \,\,\qquad||\,F\,|| \,= \displaystyle \sup _{ \Omega_T\,}\, | \,F \,(\,x,\,t,\, u) \,|, \, \]. 

In \cite{dr8} the following theorem has been proved:

\begin{theorem}
  When $ g_i\,=\,0 \, \,\,\,(i=1,2),\, $   solution (\ref{41}),  for large $ t , $  verifies  the following estimate:

 \begin{equation}   \label{44}
|u(x,t)| \, \leq \,\, 2 \,\,\bigl[\,\,||\,F\,|| \,\, \beta_0\,+\, \,  \,\,||\,u_0\,|| \,\, ( 1\, +\, \sqrt b \,\pi \,t \, ) \,\,e^{- \omega \, t\,}\, \bigr]
\end{equation}

 \noindent  where $ \, \omega = \min \,(a, \beta )$ and  $ \beta_0\,$
is defined by $(\ref{35})_{2}$.
\end{theorem}
As for contributes of boundary data, the   well known  theorem   will be considered: 
\cite{b}

\begin {theorem}
 Let  $ h(t) $ and $ \chi(t) $  be two continuous functions on $ [0,\infty [.$   If they    satisfy the following  hypotheses 

\vspace{3mm} \noindent 
\begin{equation}  \label{hp}
\exists \,\, \displaystyle{\lim_{t \to \infty}}\chi(t) \, = \, \chi(\infty)\qquad\exists \,\, \displaystyle{\lim_{t \to \infty}}h(t) \, = \, h(\infty),
\end{equation}  
 
 \noindent 
\begin{equation} \label{hp2}
 \dot  h(t)\,  \in \, L_1  [ \,0, \infty),\end{equation}

 \noindent then, it results: 

\noindent 
 \begin{equation}      \label{47}
\lim_{t \to \infty} \,\, \int_o^t \,\chi(t-\tau ) \, \dot h ( \tau ) \, d \tau \,\, = \, \,\chi(\infty) \,\, [\,\,h(\infty) - h(0)\,\,].
\end{equation} 
\end{theorem} 

\noindent According to this,  it is  possible to state:

\begin{theorem} \label{theorem asintotico}
Let  $ g_ i  \,\,\ (i=1,2) \,\,$  be two  continuous functions  converging for $ t \rightarrow \, \infty . $  In this case one has:

\begin{equation}    \label{48}
\lim_{t \to \infty } \,\int_0^t \,\theta_x \,(x,\tau)\,\,\, g_i \,(t-\tau)\, \,d\,\tau \, = \, g_{i, \infty} \,\,\,\,\,  \frac{1}{2 \, \varepsilon  }\,\,\, \frac{\sinh \sigma_0  \,\,(x-L)}{\sinh\  \sigma_0  \, L }
\end{equation}

  \noindent where  $ \sigma_0 = \sqrt{\biggl(\,a\,\,+ \dfrac{b}{\beta}\biggr) \dfrac{1}{\varepsilon}}.$
\end{theorem}

 \begin{proof}
Let us  apply (\ref{47}) with  $ h= \int _0^t \theta_x(x,\tau) d\tau \, \,\,\mbox{and}\,\, \chi = g_i  \,\,(i=1,2) $. Then, (\ref{48}) follows by $(\ref{317})_1.$
 
\end{proof}

\section{An example: estimate for the FitzHugh Nagumo system }

When  $u(x,t) $ is determined, by means (\ref{15}), the $v(x,t)\,$ component is given by

\begin{equation}      \label{52}
v\,(x,t) \, =\,v_0 \, e^{\,-\,\beta\,t\,} \,+\, b\, \int_0^t\, e^{\,-\,\beta\,(\,t-\tau\,)}\,u(x,\tau) \, d\tau.
\end{equation}
To achieve the expression of the solution $ (u,v), $ let us denote with $\, f_1 \, \ast
f_2 \, $ the convolution

\[ f_1 ( \cdot, t) \ast \, f_2 ( \cdot ,t) \, = \int_0^t \, f_1 ( \cdot, t)  \, f_2 \,\,( \cdot , t -\tau)  \,d\,\tau. \]
\noindent   So that, referring to  Dirichlet  conditions, if    

\[G(x,\xi, t) \, = \,  \theta \,(\,|x-\xi|,\, t\,)\,- \,  \theta \,(\,x+\xi,\,t\,),
 \]

\noindent and  denoting by  $ N(x,t) $  the following known function depending on the data 
$( u_0, v_0, g_1,  g_2)$:

\begin{equation}    \label{53}
N(x,t)\, =\,
  -2 \,\varepsilon \, g_1 (t) \, \ast \, \theta_x (x,t) \,+
 \end{equation}
\[ + \, 2\, \varepsilon \, \,g_2 (t) \, \ast \, \theta_x ( x-L , t) \, 
+\,\int^L_0 \,  \,u_0\,(\xi)\,\, G( x, \xi,t) \, d\xi  \, - \, e^{\,-\, \beta\, t\, } \, \ast  \, \int^L_0 \,  v_0( \xi)\,\, G( x, \xi,t) \, d\xi \,,\]
 \noindent it results:

\begin{equation}      \label{54}
\begin{split}
v\,(x,t) \, = &\,\, v_0 \, e^{\,-\,\beta\,t\,} \,+\, b\,  e^{\,-\, \beta\, t\, } \, \ast  \, \,N(x,t)
\\ \\&  +\,b\,\,\, e^{\,-\, \beta\, t\, } \, \ast  \,\int _0^L \, G \, (\, x, \xi  , t-\tau) \,\, \ast \,\varphi\,[\,\xi,\,\tau,\,u(\xi, \tau)]\, \,]\}\,\,d\xi\,. 
\end{split}
\end{equation}

So, the  asymptotic effects due to initial disturbances are vanishing, while the effects of the source terms are bounded. Indeed, letting

\[||\,u_0\,|| \,= \displaystyle \sup _{ 0\leq\,x\,\leq \,L\,}\, | \,u_0 \,(\,x\,) \,|, \,\,\qquad ||\,v_0\,|| \,= \displaystyle \sup _{ 0\leq\,x\,\leq \,L\,}\, | \,v_0 \,(\,x\,) \,|,  \, \]

\noindent and 
\[  
||\,\varphi\,|| \,= \displaystyle \sup _{ \Omega_T\,}\, | \,\varphi \,(\,x,\,t,\,u) \,|,\]

 \noindent  by means  of  (\ref{51}) (\ref{41}) and (\ref{54}) and  owing to the estimates $(\ref{33})_1$,  (\ref{37}), (\ref{38}), the following theorem holds:

\begin{theorem}
For regular solution $ (u,v) $ of the (FHN) model, when  $g_1\,\,= g_2 \,= \,0,\, \, $  the following estimates hold:

\begin{equation}            \label{55}
\left\{ 
 \begin{array}{lll}                                                   
 \left| u \, \right| \, \leq  2\,[\,\left\| u_0 \right\| \, (1+\pi \sqrt b \, t ) \, e^ {\,-\omega\,t\,}\,+\,\left\| v_0 \right\|\,E(t) \, +\, \beta_0 \,\left\| \varphi \right\|\,] 
   \\
\\
\left| v \, \right| \, \leq  \left\| v_0 \right\|\, e^ {\,-\,\beta\,t\,}\,+\,2\,[\,b\,(\,\left\| u_0 \right\|\,+\, t\, \left\| v_0 \right\|\,) \, E(t) \, + \, b\, \beta_1\, \left\| \varphi \right\| \,]
\\ 
   \end{array}
  \right.
 \end{equation}
\end{theorem}

\noindent As for the asymptotic effects of boundary perturbations $ g_1, \, g_2 \,  $   by means of  (\ref{48}), when $ u_0 =0 $ and $ F=0, $ one has  

\begin{equation}            \label{56}
\left\{ 
 \begin{array}{lll}                                                   
 u  \, = g_{1,\infty}\,\,\frac{\sinh \sigma_0  \,\,(L-x)}{\sinh\  \sigma_0  \, L }  + g_{2,\infty}\,\,\frac{\sinh \sigma_0  \, x}{\sinh\  \sigma_0  \, L } \big| 
   \\
\\
 v \, =  \, \,\dfrac{b}{\beta} \,\,\bigg[\,\,g_{1, \infty} \,\,\,  \,\, \frac{\sinh \sigma_0  \,\,(L-x)}{\sinh\  \sigma_0  \, L } \,\,+ \,\,g_{2, \infty} \,\,\,  \,\, \frac{\sinh \sigma_0  \,\,(x)}{\sinh\  \sigma_0  \, L } \,\bigg].
\\ 
   \end{array}
  \right.
 \end{equation}

\section{Remarks}

\hspace{5mm}$ \bullet $ The paper is concerned with  the nonlinear integral equation (\ref{11}) whose kernel is a Green function with  numerous basic properties typical of the diffusion equation.

$ \bullet $ Neumann, Dirichlet and mixed boundary conditions are considered, and   integro differential formulations  of  {\em non linear} problems are obtained.

$ \bullet $ The asymptotic behavior for initial boundary value problem with  Dirichlet conditions is  evaluated, showing that effects  due to initial disturbances vanish, while the influences of the source term  and boundary perturbations are everywhere  bounded.

$ \bullet $The analysis  related to  Dirichlet conditions can be applied to  mixed problem, too. Indeed, like $ \theta(x,t), $ also  the Green function  $ \theta^*(x,t)  $ defined in (\ref{27})  depends on   the fundamental solution  $ K_0. $

$ \bullet $ The equivalence among equation   (\ref{11})  and  numerous models   allow us to apply asymptotic theorems to many  other problems related to various physical fields. 

\section*{Acknowledgment}
This  work has been performed under the auspices of G.N.F.M. of I.N.d.A.M. and of Programma F.A.R.O. (Finanziamenti per l' Avvio di  Ricerche
Originali, III tornata) ``Controllo e stabilita' di processi diffusivi nell'ambiente'', Polo delle Scienze e Tecnologie, Universita' degli Studi di Napoli Federico II  (2012).

\label{Sample_NDST: LastPage}

\end{document}